\documentclass[twocolumn,preprintnumbers,amsmath,amssymb]{revtex4}

\usepackage{array}
\usepackage{booktabs}
\usepackage{tabu}
\usepackage{dcolumn}
\usepackage{amsmath}
\usepackage{amsfonts}
\usepackage{amssymb}
\usepackage{graphicx}
\usepackage{subfigure}
\usepackage{graphicx}
\usepackage{dcolumn}
\usepackage{bm}
\usepackage{xcolor}
\usepackage{amsthm}
\def\be{\begin{equation}}
\def\ee{\end{equation}}
\def\bea{\begin{eqnarray}}
\def\eea{\end{eqnarray}}
\def\f{\frac}
\def\n{\nonumber}
\def\l{\label}
\def\p{\phi}
\def\o{\over}
\def\R{\rho}
\def\pa{\partial}
\def\om{\omega}
\def\na{\nabla}
\def\P{\Phi}
\newtheorem{theorem}{Theorem}

\begin{document}

\title{Tightening the tripartite quantum memory assisted entropic uncertainty relation}

\author{H. Dolatkhah}
\affiliation{Department of Physics, University of Kurdistan, P.O.Box 66177-15175, Sanandaj, Iran}
\author{S. Haseli}
\email{soroush.haseli@uut.ac.ir}
\affiliation{Faculty of Physics, Urmia University of Technology, Urmia, Iran}
\author{S. Salimi}
\author{A. S. Khorashad}
\affiliation{
Department of Physics, University of Kurdistan, P.O.Box 66177-15175, Sanandaj, Iran.\\}
\date{\today}

\def\be{\begin{equation}}
  \def\ee{\end{equation}}
\def\bea{\begin{eqnarray}}
\def\eea{\end{eqnarray}}
\def\f{\frac}
\def\n{\nonumber}
\def\l{\label}
\def\p{\phi}
\def\o{\over}
\def\R{\rho}
\def\pa{\partial}
\def\om{\omega}
\def\na{\nabla}
\def\P{\Phi}

\begin{abstract}
The uncertainty principle determines the distinction  between the classical and quantum worlds. This principle states that it is not possible to measure two incompatible observables with the desired accuracy simultaneously. In quantum information theory, Shannon entropy has been used as an appropriate measure to express the uncertainty relation. According to the applications of entropic uncertainty relation, studying and trying to improve the bound of this relation is of great importance. Uncertainty bound can be altered by considering an extra quantum system as the quantum memory $B$ which is correlated with the measured quantum system $A$. One can extend the bipartite quantum memory assisted entropic uncertainty relation to tripartite quantum memory assisted entropic uncertainty relation in which the memory is split into two parts. In this work, we obtain a lower bound for the tripartite quantum memory assisted entropic uncertainty relation. Our lower bound has two additional terms compared to the lower bound in [Phys. Rev. Lett. 103, 020402 (2009)] which depending on the conditional von Neumann entropy, the Holevo quantity and mutual information. It is shown that the bound obtained in this work is more tighter than other bounds. In addition, using our lower bound, a lower bound for the quantum secret key rate has been obtained. The lower bound is also used to obtain the states for which the strong subadditivity inequality and Koashi-Winter inequality is satisfied with equality.
\end{abstract}
\maketitle

\section{Introduction}	
Uncertainty principle is undoubtedly one of the fundamental concepts in quantum theory. This principle defines the distinction between the classical  and the quantum world \cite{Heisenberg}. This principle sets a bound on our ability to predict the measurement outcomes of two incompatible observables which simultaneously are measured on a quantum system. Actually, for two arbitrary observables, $X$ and $Z$, Robertson showed that \cite{Robertson,Schorodinger}
\begin{equation}\label{Eq1}
\Delta X \Delta Z\geq\frac{1}{2}|\langle [X, Z]\rangle|,
\end{equation}
where $\Delta P=\sqrt{\langle P^{2}\rangle-\langle P \rangle^{2}} $  with $P \in \{X, Z\}$ shows the standard deviation, $\langle P \rangle$ represents  the expectation value of operator $\kappa$ , and $[X, Z] =XZ-ZX$. This form of uncertainty relation is still one of the most well-known uncertainty relations.\\
In quantum information theory, the uncertainty principle can be formulated in terms of the Shannon entropy. The most famous form  of entropic uncertainty relation (EUR) was introduced by Deutsch  \cite{Deutsch} and then improved by Massen and Uffink  \cite{Uffink}.  They have shown that for two incompatible observables $X$ and $Z$, the following EUR holds
\begin{equation}\label{Maassen and Uffink}
H(X)+H(Z)\geqslant \log_2 \frac{1}{c}\equiv q_{MU},
\end{equation}
where $H(P) = -\sum_{k} p_k \log_2 p_k$ is the Shannon entropy of the measured observable $P \in \lbrace X, Z \rbrace$, $p_k$ is the probability of the outcome $k$, the quantity $c$ is defined as $c = \max_{\lbrace \mathbb{X},\mathbb{Z}\rbrace } \vert\langle x_{i} \vert z_{j}\rangle \vert ^{2}$, where  $\mathbb{X}=\lbrace \vert x_{i}\rangle \rbrace$ and $\mathbb{Z}=\lbrace \vert z_{j}\rangle \rbrace$ are eigenstates of observables $X$ and $Z$, respectively and $q_{MU}$ is called incompatibility measure. \\
The EUR has a wide range of different applications in the field of quantum information, including quantum key distribution \cite{Koashi,Berta}, quantum cryptography \cite{Dupuis,Koenig}, quantum randomness \cite{Vallone,Cao}, entanglement witness \cite{Berta2,Huang,Bagchi}, EPR steering \cite{Walborn,Schneeloch}, and quantum metrology \cite{Giovannetti}.\\
So far, many efforts have been made to expand and modify this relation \cite{Berta,Coles1,Bialynicki,Pati,Ballester,Vi,Wu,Wehner,Rudnicki,Rudnicki1,Pramanik,Maccone,Pramanik1, Zozor,Coles,Adabi,Adabi1,Dolatkhah,Haseli2,Yunlong,Liu,Kamil,Zhang,R}.  Berta \emph{et al.} studied bipartite quantum memory assisted entropic uncertainty relation (QMA-EUR) \cite{Berta} which can be explained by means of an interesting game between two players, Alice and Bob. At the beginning of the game, Alice and Bob share a quantum state $\rho_{AB}$. In the next step, Alice carries out  a measurement on her quantum system $A$ by choosing one of the observables $X$ and $Z$, then she  announces her choice of the measurement to Bob which keeps  the quantum memory $B$. Bob's task is to predict the outcome of Alice's measurement. It is shown that the bipartite QMA-EUR can be written as \cite{Berta}  
\begin{equation}\label{Berta}
S(X \vert B)+S(Z \vert B) \geqslant q_{MU} +S(A \vert B),
\end{equation}
where $S(P \vert B) = S(\rho_{PB})-S(\rho_{B})$ $(P \in \lbrace X, Z \rbrace)$ are the conditional von Neumann entropies of the post measurement states after measuring $X$ or $Z$ on the part $A$, 
\begin{equation}
\rho_{XB}= \sum_{i}(\vert x_{i}\rangle\langle x_{i}\vert_{A}\otimes \mathbf{I}_B ) \rho_{AB}(\vert x_{i}\rangle\langle x_{i}\vert_{A}\otimes \mathbf{I}_{B} ),\nonumber
\end{equation}
 
\begin{equation}
\rho_{ZB}= \sum_{j}(\vert z_{j}\rangle\langle z_{j}\vert_{A}\otimes \mathbf{I}_{B} ) \rho_{AB}(\vert z_{j}\rangle\langle z_{j}\vert_{A}\otimes \mathbf{I}_{B} ),\nonumber
\end{equation}
 and $S(A|B) = S(\rho_{AB})-S(\rho_{B})$ is  the conditional von Neumann entropy. Note that when the conditional entropy $S(A \vert B)$ is negative which means that $A$ and $B$  are entangled, Bob can predict Alice's measurement outcomes with better accuracy. Moreover, when the measured particle $A$ and the memory particle $B$ are maximally entangled, Bob can perfectly predict Alice's measurement outcomes. Also, in the absence of a quantum memory, Eq.\;(\ref{Berta}) reduces to
\begin{equation}\label{Berta2}
H(X)+H(Z)\geqslant q_{MU}+S(A),
\end{equation}
which is tighter than the Maassen and Uffink EUR due to $S(A)\geqslant 0$.

Much efforts has been made to improve the lower bound of the bipartite QMA-EUR \cite{Pati,Adabi,Coles}. Pati \emph{et al.} improved the Berta's bound by adding a term to the lower bound in Eq.\;(\ref{Berta}). The term depends on the classical correlation and quantum discord \cite{Pati}.

Adabi \emph{et al.} provided a lower bound for the uncertainties $S(X \vert B)$ and
$S(Z \vert B)$ by considering an additional term on the right-hand side of Eq.\;(\ref{Berta}) \cite{Adabi},
\begin{equation}\label{new1}
S(X \vert B)+S(Z \vert B)\geqslant q_{MU} + S(A|B)+\max\{0 , \delta\},
\end{equation}
where  $$\delta=I(A:B)-[I(X:B)+I(Z:B)],$$
in which $$I(A:B)=S(\rho_{A})+S(\rho_{B})-S(\rho_{AB})$$  is mutual information
and 
$$I(P:B)= S(\rho_{B})- \sum_{i}p_{i}S(\rho_{B|i})$$
is the Holevo quantity. It is equal to the upper bound of the accessible information to Bob about Alice's measurement outcomes. Note that when Alice measures the observable $P$ on the part A, the $i$-th outcome with probability $p_{i}= Tr_{AB}(\Pi^{A}_{i}\rho_{AB}\Pi^{A}_{i})$ is obtained and the part $B$ is left in the corresponding state $\rho_{B|i}= \frac{Tr_{A}(\Pi^{A}_{i}\rho_{AB}\Pi^{A}_{i})}{p_{i}}$. Adabi \emph{et al.} showed that this lower bound is tighter than both the Berta and Pati lower bounds.\\
It is possible to extend the bipartite QMA-EUR to tripartite one. In tripartite scenario, two additional particle $B$ and $C$ are considered as the quantum memories. In fact, parts $A$, $B$, and $C$ are available to Alice, Bob, and Charlie, respectively. In this case, Alice, Bob, and Charlie share a quantum state $\rho_{ABC}$ and Alice carries out one of two measurements, $X$ and $Z$, on her quantum system. If she measures $X$, then Bob's task is to minimize his uncertainty about $X$.  If she measures $Z$, then Charlie’s task is to minimize his uncertainty about $Z$. It is shown that the tripartite QMA-EUR can be expressed as
\begin{equation}\label{tpu}
S(X \vert B)+S(Z \vert C)\geqslant q_{MU},
\end{equation}
where $q_{MU}$ is the same as that in Eq.\;(\ref{Berta}). This equation was conjectured by Renes and Boileau \cite{Renes} and then proved by Berta \emph{et al.} \cite{Berta}. The proof was simplified by Tomamichel and Renner \cite{Tomamichel}  and Coles \emph{et al.} \cite{Coles44}.
 
As can be seen the lower bound of Eq.\;(\ref{tpu}) depends only on the complementarity of the observables. This means that with respect to given observables $X$ and $Z$, this  lower bound is a constatnt value for any state $\rho_{ABC}$ that is shared between Alice, Bob, and Charlie. Although the bound should be depends on the state of system. According to our knowledge so far, there have been few improvement of the tripartite QMA-EURs. However, recently Ming  \emph{et al.} \cite{Ming} improved the lower bound of the tripartite QMA-EUR by adding a term to the lower bound in Eq.\;(\ref{tpu}),
 \begin{equation}\label{tpu2}
S(X \vert B)+S(Z \vert C)\geqslant q_{MU}+\max\{0 , \Delta\},
\end{equation}
where 
\begin{eqnarray}
\Delta &=&q_{MU}+2S(A)-\left[I(A:B)+I(A:C)\right]\nonumber \\
&+&\left[I(Z:B)+I(X:C)\right]-H(X)-H(Z).
\end{eqnarray}
They showed that this lower bound is tighter than that of Eq.\;(\ref{tpu}).

In this work,  we introduce a lower bound for the tripartite QMA-EUR by adding two additional terms  to the lower bound in Eq. (\ref{tpu}) which one depending on the conditional von Neumann entropies and the other one depending on the mutual information and the Holevo quantity. We show that there exist states $\rho_{ABC}$ for which our lower bound is perfectly tight.
We examine our lower bound for four examples and compare our lower bound with the other lower bounds. We show that our lower bound for generalized Greenberger-Horne-Zeilinger (GHZ) states and Werner-type states coincides with Ming \emph{et al.} lower bound \cite{Ming} and for generalized W states and symmetric family of mixed three qubit states our lower bound is tighter than that of Ming \emph{et al.}. As applications, here we obtain a lower bound for the quantum secret key rate based on our results. Also, we explain how our lower bound can be applied to find states that saturating the strong subadditivity (SSA) inequality and Koashi-Winter inequality. 

The paper is organized as follows: In Sec. \ref{Sec2} we introduce the lower bound for the tripartite QMA-EUR. In Sec. \ref{Sec3} we examine our lower bound for four examples and compare our lower bound with the other lower bounds. In Sec. \ref{Sec4}, we discuss two of the applications of of our lower bound. Finally, the results are summarized in Sec. \ref{conclusion}.

\section{Improved tripartite QMA-EUR}\label{Sec2}
 In this section, a lower bound for  the tripartite QMA-EUR is obtained.\\ 
\begin{theorem}
Let $X$ and $Z$ be two incompatible observables with bases $\mathbb{X}$ and $\mathbb{Z}$, respectively. The following tripartite uncertainty relation holds for any state $\rho_{ABC}$,
\begin{equation}\label{tpu3}
S(X \vert B)+S(Z \vert C)\geqslant q_{MU}+\frac{S(A|B)+S(A|C)}{2}+\max\{0 , \delta\},
\end{equation}
where  $$\delta=\frac{I(A:B)+I(A:C)}{2} -[I(X:B)+I(Z:C)],$$
\end{theorem}
\begin{proof}
Regarding $S(X\vert B)=H(X)-I(X:B)$ and $S(Z\vert C)=H(Z)-I(Z:C)$, the left-hand side of Eq.\;(\ref{tpu3}) can be rewritten as\

\begin{eqnarray}\label{p1}
S(X \vert B)+S(Z \vert C)&= &  H(X)+H(Z)-I(X:B)-I(Z:C) \nonumber \\ 
&\geqslant & q_{MU}+S(A)-I(X:B)-I(Z:C) \nonumber \\
&=& q_{MU}+S(A|B) \nonumber \\
&+&I(A:B)-[I(X:B)+I(Z:C)],
 \end{eqnarray}
where the inequality  follows from the Eq.\;(\ref{Berta2}) and last equality comes from the identity $S(A)=S(A|B)+I(A:B)$.
Using $S(A)=S(A|C)+I(A:C)$ in the last line of the above equation, one can obtain
\begin{eqnarray}\label{p2}
S(X \vert B)+S(Z \vert C)& \geqslant & q_{MU}+S(A|C)  \\
&+&I(A:C)-[I(X:B)+I(Z:C)],\nonumber
\end{eqnarray}
from Eqs.\;(\ref{p1}) and (\ref{p2}), one arrives at 
\begin{eqnarray}
S(X \vert B)+S(Z \vert C)&\geqslant & q_{MU}+\frac{S(A|B)+S(A|C)}{2}\nonumber \\
&+&\frac{I(A:B)+I(A:C)}{2} \nonumber \\
&-&[I(X:B)+I(Z:C)],
\end{eqnarray}
which can be rewritten as:
\begin{equation}
S(X \vert B)+S(Z \vert C)\geqslant q_{MU}+\frac{S(A|B)+S(A|C)}{2}+\max\{0 , \delta\}.
\end{equation}
\end{proof}
 
As can be seen the lower bound of Eq.\;(\ref{tpu3}) includes three terms. The first, $q_{MU}$, 
 depends on the complementarity of the observables. The second, $\frac{S(A|B)+S(A|C)}{2}$, is closely related to the SSA inequality. The third, $\max\{0 , \delta\}$, depends on mutual information and Holevo quantity.\\   
It should be mentioned that the inequality $\frac{S(A|B)+S(A|C)}{2} \geqslant 0$ is always true due to the SSA inequality \cite{Nielsen}. Thus, it is clear that our lower bound is stronger than Eq.\;(\ref{tpu})  since the additional terms, $\frac{S(A|B)+S(A|C)}{2} \geqslant 0$ and $\max\{0 , \delta\}$ are non-negative.\\

As mentioned in the introduction, in the bipartite QMA-EUR, entanglement between the part $A$ and the quantum memory $B$, $S(A|B)<0$, leads to Bob can predict Alice’s measurement outcomes of both observables $X$ and $Z$ with better accuracy. While in the tripartite QMA-EUR, if the conditional entropy
 $S(A|B)$ is negative, then the conditional entropy $S(A|C)$ must be positive to satisfy the SSA inequality, making it hard for Charlie to guess  Alice’s measurement outcomes of observable $Z$. In other words, Eq.\;(\ref{tpu3}) states that the more Bob knows about $X$, the less Charlie knows about $Z$, and vice versa. \\
It is worth noting that there are some special cases for which our lower bound is perfectly tight. One is that if $X$ and $Z$ are complementary and the subsystem $A$ is maximally mixed, such as the GHZ state. In another case, when $X$ and $Z$ are complementary and $X$ $(Z)$ minimally disturbs subsystem $A$, $H(X)$ $[H(Z)]$ is equal to $S(A)$, and $H(Z)$  $[H(X)]$ is equal to $log_{2}d$, which leads to
\begin{equation}
H(X)+H(Z)=log_{2}d+S(A),
\end{equation} 
where $d$ is the dimension of subsystem $A$. In this case again, our lower bound in Eq.\;(\ref{tpu3}) is extremely tight, such as the generalized GHZ states, the generalized W states, the  Werner-type states and the three-qubit $X$-structure states.

\section{Examples}\label{Sec3}

\subsection{Generalized GHZ state}\label{GHZ}
First, let us consider the  generalized GHZ states defined as
\begin{equation}\label{GHZ}
\vert GGHZ\rangle = cos\beta \vert 000\rangle +sin\beta \vert 111\rangle,  
\end{equation}
where $\beta\in\left[  0,2\pi\right)$. Two complementary observables measured on the part $A$ of this state are assumed to be the Pauli matrices, $X = \sigma_{1}$ and $Z = \sigma_{3}$. In Fig. \ref{fig1}, different lower bounds of the tripartite QMA-EUR for these states are plotted versus the parameter $\beta$. As can be seen, our lower bound coincides with those in Eqs. (\ref{tpu}) and (\ref{tpu2}).
\begin{figure}[ht] 
\centering
\includegraphics[width=8cm]{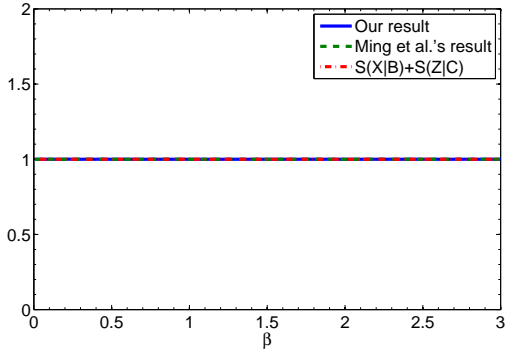}
\caption{(Color online) Different lower bounds of the tripartite QMA-EUR for two complementary observables $X=\sigma_{1}$ and $Z=\sigma_{3}$ measured on the part $A$ of the state in Eq.\;(\ref{GHZ}), versus the parameter $\beta$.}\label{fig1}
\end{figure}
\begin{figure}[ht] 
\centering
\includegraphics[width=8cm]{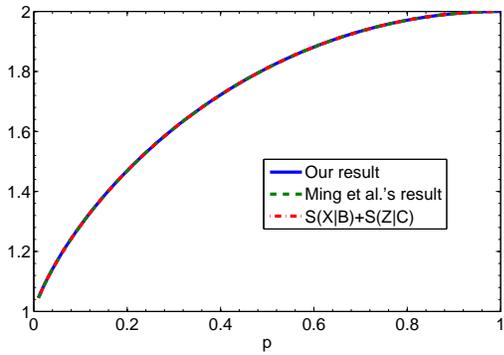}
\caption{(Color online) Different lower bounds of the tripartite QMA-EUR for two complementary observables $X=\sigma_{1}$ and $Z=\sigma_{3}$ measured on the part $A$ of the state in Eq.\;(\ref{werner}), versus the parameter $p$.}\label{fig2}
\end{figure}

\subsection{Werner-type state}

As a second example, let us consider the  Werner-type states defined as
 \begin{equation}\label{werner}
 \rho_{w}=(1-p) \vert GHZ \rangle \langle GHZ \vert + \frac{p}{8}\mathbf{I}_{ABC},
 \end{equation}
 where $\vert GHZ \rangle = 1/\sqrt{2}(\vert 000 \rangle + \vert 111 \rangle)$ is the GHZ state  and $0 \leq p \leq 1$. For these states we find that our lower bound and Ming \emph{et al.}’s lower bound completely coincide with each other, as shown in Fig. \ref{fig2}.
\begin{figure}[ht] 
\centering
\includegraphics[width=8cm]{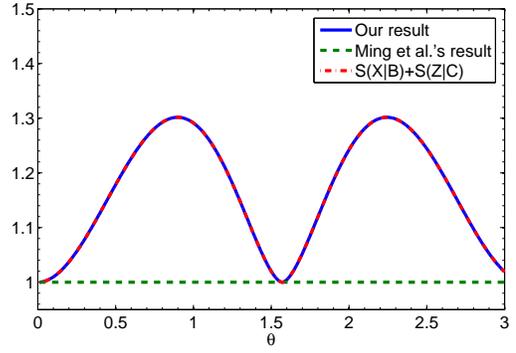}
\caption{(Color online) Different lower bounds of the tripartite QMA-EUR for two complementary observables $X=\sigma_{1}$ and $Z=\sigma_{3}$ measured on the part $A$ of the state in Eq.\;(\ref{W}), versus the parameter $\theta$, where $\phi=\pi/4$}\label{fig3}
\end{figure}
\begin{figure}[ht] 
\centering
\includegraphics[width=8cm]{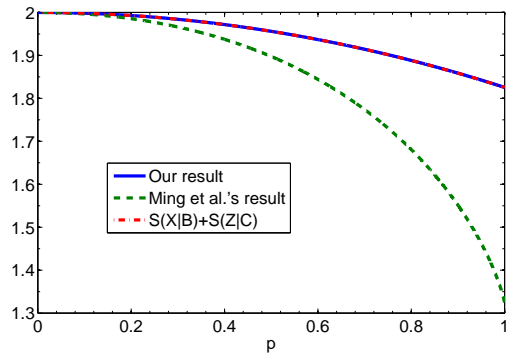}
\caption{(Color online) Different lower bounds of the tripartite QMA-EUR for two complementary observables $X=\sigma_{1}$ and $Z=\sigma_{3}$ measured on the part $A$ of the state in Eq.\;(\ref{Mix}), versus the parameter $p$, where $ 0 \leq p \leq 1$}\label{fig4}
\end{figure}
\subsection{Generalized W state}
As an another example, we consider the generalized W states defined as
\begin{equation}\label{W}
\vert GW \rangle = sin\theta cos\phi \vert 100\rangle +sin\theta sin\phi \vert 010\rangle +cos\theta \vert 001\rangle ,
\end{equation}
where $\theta\in\left[  0,\pi\right]$ and $\phi\in\left[  0,2\pi\right)$. In Fig. \ref{fig3}, the lower bounds of the tripartite QMA-EUR  for these states are plotted versus the parameter $\theta$. As can be seen,  Ming \emph{et al.}’s lower bound remain a constant
in all intervals related to parameter $\theta$ whereas our lower bound is extremely tight.

\subsection{A symmetric family of mixed three qubit states}
As the last example, let us consider the mixture of the GHZ
state, the W state, and the maximally mixed three-qubit state
\begin{equation}\label{Mix}
\rho=\frac{1-p}{8}\mathbf{I}_{ABC}+\frac{p}{2}\vert GHZ \rangle \langle GHZ \vert + \frac{p}{2}\vert W \rangle \langle W \vert,
\end{equation}
where $ 0 \leq p \leq 1$ is real number and the usual $\vert GHZ \rangle$ and $\vert W \rangle$ states are defined as 
\begin{eqnarray}
\vert GHZ \rangle &=& \frac{1}{\sqrt{2}}(\vert 000 \rangle + \vert 111 \rangle),  \\
\vert W \rangle &=& \frac{1}{3}(\vert 001 \rangle + \vert 010 \rangle + \vert 100 \rangle). \nonumber 
\end{eqnarray} 

In Fig. \ref{fig4} the lower bounds of the tripartite QMA-EUR  for symmetric family of mixed three qubit states are plotted versus the parameter $p$. As can be seen, our lower bound is tighter than that of Ming \emph{et al.} and is identical with the sum of Bob's and Charlie's uncertainties $S(X \vert B)+S(Z \vert C)$.

\section{ Applicatin}\label{Sec4}
\subsection{Quantum key distribution}

The main purpose of the key distribution protocol is the agreement on a shared key between two honest parts (Alice and Bob) by communicating over a public channel in a way that the key is secret from any eavesdropping by the third part (Eve). It has been shown that the amount of the key $K$ that can be extracted by Alice and Bob is lower bounded by \cite{Devetak} 
\begin{equation}\label{app1}
K \geqslant S(Z \vert E) - S(Z \vert B),
\end{equation}
where the eavesdropper (Eve) prepares a quantum state $\rho_{ABE}$ and sends the parts $A$ and $B$ to Alice and Bob, respectively, and keeps $E$.
 
Note that the lower bound of the tripartite QMA-EUR is closely connected with the quantum secret key (QSK) rate. Eq. (\ref{tpu}), leads us to
\begin{equation}\label{app2}
S(Z \vert E) \geqslant q_{MU} - S(X \vert B).
\end{equation}

Regarding  Eqs.\;(\ref{app1}) and (\ref{app2}), Berta \emph{et al.} obtained the following relation for the bound of the QSK rate  \cite{Berta}
\begin{equation}\label{keyrateb}
K \geqslant q_{MU} - S(X \vert B)-S(Z \vert B).
\end{equation} 
Using Eq.\;(\ref{tpu3}), one can obtain a new lower bound on the QSK rate which is 
\begin{eqnarray}\label{keyrateH}
K^{\prime} & \geqslant & q_{MU}+\frac{S(A|B)+S(A|C)}{2}  \nonumber \\
&+&\max\{0 , \delta\}- S(X \vert B)-S(Z \vert B).
\end{eqnarray} 
Compared with Eq.\;(\ref{keyrateb}), the QSK rate has lower bounded by two additional terms in Eq.\;(\ref{keyrateH}). Since these terms are greater than or equal
to zero, one comes to the result  that $K^{\prime}$ is tighter than $K$. 
\subsection{Strong subadditivity}
The SSA inequality states that
\begin{equation}\label{SSA}
S(\rho_{ABC})+S(\rho_{C})\leqslant S(\rho_{AC})+S(\rho_{BC}).
\end{equation}
Eq.\;(\ref{SSA}) is equivalent to
\begin{equation}
\frac{S(A|B)+S(A|C)}{2}\geqslant 0.
\end{equation}
This  form of SSA states that although the quantum conditional entropies can be negative, both of them cannot be negative simultaneously. The structure of states for which the SSA inequality is saturated is not trivial \cite{Petz,Ruskai,Hayden}. According to our relation if $S(X \vert B)+S(Z \vert C)- q_{MU}-\max\{0 , \delta\}=0,$ then $\frac{S(A|B)+S(A|C)}{2}=0$, which means that $\rho_{ABC}$ satisfies the SSA inequality with equality. Moreover, due to the fact that measurement do not decrease entropy and $\max\{0 , \delta\}\geqslant 0$, one arrives at if $S(X \vert X^{\prime})+S(Z \vert Z^{\prime})- q_{MU}=0,$ then $\rho_{ABC}$ satisfies the SSA inequality with equality.
 It is also shown that if $\rho_{ABC}$ satisfies the SSA inequality with equality, then \cite{Akhtarshenas} 
\begin{enumerate} 
\item  it satisfies the Koashi-Winter relation \cite{Koashi1} with equality, 
 \begin{equation}\label{Koashi-Winter}
E(\rho_{AB})=D^{C}(\rho_{AC})+S(A|C),
 \end{equation}
 and
 \begin{equation}\label{Koashi-Winter1}
  E(\rho_{AC})=D^{B}(\rho_{AB})+S(A|B),
 \end{equation}
\item it satisfies the quantum conservation law \cite{Fanchini1},
\begin{equation}\label{conservation law}
E(\rho_{AB})+E(\rho_{AC})=D^{B}(\rho_{AB})+D^{C}(\rho_{AC}),
\end{equation}
where $E(\rho_{AY})$, $Y=B, C$, is entanglement of
formation which defined as
\begin{equation}
E(\rho_{AY})=\min_{\lbrace p_{i},\vert \psi_{i} \rangle\rbrace } \sum_{i}p_{i}S(Tr_{Y}(\vert \psi_{i} \rangle \langle \psi_{i} \vert)), 
\end{equation}  
in which minimum is taken over all ensembles ${\lbrace p_{i},\vert \psi_{i} \rangle\rbrace }$ satisfying $\rho_{AY}= \sum_{i}p_{i}\vert \psi_{i} \rangle$, and $D^{Y}(\rho_{AY}):=I(A:B)-J_{Y}(\rho_{AY})$ is quantum discord, where 
\begin{equation}\label{ss}
J_{Y}(\rho_{AY})=\max_{\lbrace \Pi_{i}^{Y} \rbrace}I(A:P)
\end{equation} 
is the classical correlation of the state $\rho_{AY}$. The maximization is over all set of  projection operators ${\lbrace \Pi_{i}^{Y} \rbrace}$ acting on the subsystem $Y$. 
\end{enumerate}

 Therefore, based on the abov-mentioned, one can conclude that $S(X \vert B)+S(Z \vert C)- q_{MU}-\max\{0 , \delta\}=0,$ implies Eqs.\;(\ref{Koashi-Winter}), (\ref{Koashi-Winter1}) and (\ref{conservation law}). Also, $S(X \vert X^{\prime})+S(Z \vert Z^{\prime})- q_{MU}=0,$ implies Eqs.\;(\ref{Koashi-Winter}), (\ref{Koashi-Winter1}) and (\ref{conservation law}). 

\section{Conclusion}\label{conclusion}
 In this work, we have obtained a lower bound for the tripartite QMA-EUR by adding two additional terms depending on the conditional von Neumann entropy, the Holevo quantity and mutual information. We have showed that there are some special cases for which our lower bound is perfectly tight. We have compared our lower bound with the other lower bounds for some examples: especially, for the generalized W states and symmetric family of mixed three qubit states, the comparison of the lower bounds is depicted in Figs. \ref{fig3} and \ref{fig4}, where it is clear that our lower bound is tighter than that of Ming \emph{et al.}. Regarding the tripartite QMA-EUR, we could derive a lower bound for the quantum secret key rate. Also, we have explained how our lower bound can be applied to obtain the states for which the strong subadditivity inequality and Koashi-Winter inequality is satisfied with equality.










\end{document}